\documentclass[11pt]{article}
\usepackage{amsfonts,amsthm,amssymb}
\usepackage[fleqn]{amsmath}
\usepackage{url}
\usepackage[noadjust]{cite}
\usepackage{enumitem}
\usepackage[leftcaption]{sidecap}
\theoremstyle{plain}
\newtheorem{theorem}{Theorem}[section]
\newtheorem{lemma}[theorem]{Lemma}
\newtheorem{proposition}[theorem]{Proposition}
\newtheorem{corollary}[theorem]{Corollary}

\theoremstyle{definition}
\newtheorem{definition}[theorem]{Definition}
\newtheorem{remark}[theorem]{Remark}
\newtheorem{example}[theorem]{Example}

\numberwithin{equation}{section}
\numberwithin{theorem}{section}

\allowdisplaybreaks

\begin{document}

\title{Differential Calculi on Associative Algebras and Integrable Systems}

\author{ \sc{Aristophanes Dimakis}$^a$ and \sc{Folkert M\"uller-Hoissen}$^{b,c}$ \\
 \small
 $^a$ Dept. of Financial and Management Engineering,
 University of the Aegean, 82100 Chios, Greece \\  \small e-mail: dimakis@aegean.gr \\
  \small
 $^b$ Max Planck Institute for Dynamics and Self-Organization,
        37077 G\"ottingen, Germany \\
 \small       
 $^c$ Institute for Theoretical Physics, University of G\"ottingen, 
      37077 G\"ottingen, Germany \\
    \small   e-mail: folkert.mueller-hoissen@theorie.physik.uni-goettingen.de
}

\date{ }
\maketitle

\abstract{
After an introduction to some aspects of bidifferential calculus on associative algebras, we focus on the 
notion of a ``symmetry'' of a generalized zero curvature equation and derive B\"acklund and (forward, 
backward and binary) Darboux transformations from it. 
We also recall a matrix version of the binary Darboux transformation and, inspired by the so-called 
Cauchy matrix\index{Cauchy matrix} 
approach, present an infinite system of equations solved by it. Finally, we sketch recent work on a deformation 
of the matrix binary Darboux transformation in bidifferential calculus, leading to a treatment of integrable 
equations with sources. 
}


\section{Introduction}
\label{sec:intro_DMH}
Integrability\index{integrability} of a partial differential and/or difference equation (PDDE\index{PDDE}), or more generally 
a system of such equations, is most frequently understood in the sense that the equations 
arise as the compatibility (or integrability) condition of a set of linear equations, a 
Lax system\index{Lax system}. 
For some integrable PDEs in two dimensions, the compatibility condition can be expressed as 
the condition of vanishing curvature\index{curvature} of a connection, but this differential-geometric framework 
is too narrow. However, there is a generalization of the differential-geometric setting 
to a framework of ``noncommutative geometry''\index{geometry! noncommutative}, preserving simple computational rules. 
A crucial point is the use of ``generalized differential forms''\index{differential form! generalized} 
instead of ``generalized vector fields'' 
or generalized derivations\index{derivation! generalized} (see, e.g., \cite{March88_DMH,Bout+Marc94_DMH}) as the basic structure. 
Whereas, for example, discrete derivatives obey a modified derivation rule, on the level of forms 
one can preserve the simple graded derivation\index{derivation! graded} (Leibniz) rule for a 
generalized exterior derivative\index{exterior derivative! generalized}. 
Substantial results can then be derived by very simple and universal computations. 

Let $\mathcal{A}$ be an associative algebra\index{algebra! associative}, over a field $\mathbb{K}$ of characteristic zero, 
and $(\boldsymbol{\Omega}, \delta)$ a differential calculus\index{differential calculus} over $\mathcal{A}$. Here  
$\boldsymbol{\Omega} = \bigoplus_{k \geq 0} \boldsymbol{\Omega}^k$ is a graded algebra\index{algebra! graded} with 
$\boldsymbol{\Omega}^0 = \mathcal{A}$ and $\mathcal{A}$-bimodules 
$\boldsymbol{\Omega}^k$, and $\delta$ is a derivation of degree one\index{derivation! of degree one} satisfying $\delta^2 = 0$.  
A connection\index{connection} on a left $\mathcal{A}$-module\index{module} $\Gamma$ is a linear map $\nabla : \Gamma \to \boldsymbol{\Omega}^1 \otimes_\mathcal{A} \Gamma$,
such that $\nabla(f \gamma) = \delta f \otimes_\mathcal{A} \gamma + f \nabla \gamma$, for all $f \in \mathcal{A}$ 
and $\gamma \in \Gamma$. It extends to $\boldsymbol{\Omega} \otimes_\mathcal{A} \Gamma$, so that we can define the 
curvature\index{curvature} as $\nabla^2$. The essence of integrability\index{integrability} of some ``condition'' 
may then be expressed as
\begin{eqnarray*}
      \nabla^2=0 \quad \Longleftrightarrow \quad \mbox{condition}  \, .
\end{eqnarray*}
For example, here ``condition'' could mean that a variable solves a certain PDDE\index{PDDE}. In this case we are looking for 
a differential calculus\index{differential calculus} and a connection\index{connection} on a module\index{module}, depending 
on the dependent variable of the PDDE, such that the zero curvature\index{curvature} condition holds iff 
the variable satisfies the PDDE. 

We will make the simplifying assumption that the module\index{module} $\Gamma$ has a basis $b^\mu$, $\mu=1,\ldots,m$. Then, using 
the summation convention, we have $\gamma = \gamma_\mu \, b^\mu$ and
\begin{eqnarray*}
  \nabla \gamma &=& (\delta \gamma_\mu + \gamma_\nu \, \boldsymbol{A}^\nu{}_\mu) \otimes_\mathcal{A} b^\nu \, , \\
  \nabla^2 \gamma &=& \gamma_\mu \, F_\delta[\boldsymbol{A}]^\mu{}_\nu \otimes_\mathcal{A} b^\nu \, , \quad 
  F_\delta[\boldsymbol{A}]^\mu{}_\nu = \delta \boldsymbol{A}^\mu{}_\nu - \boldsymbol{A}^\mu{}_\kappa \, \boldsymbol{A}^\kappa{}_\nu \, ,
\end{eqnarray*}
where $\boldsymbol{A}$ is the $m \times m$ matrix of elements of $\boldsymbol{\Omega}^1$ defined by 
$\nabla b^\mu = \boldsymbol{A}^\mu{}_\nu \otimes_\mathcal{A} b^\nu$.  

However, in most cases the above characterization of integrability\index{integrability} is not strong enough. Rather, one 
needs a connection\index{connection} that depends on a (``spectral'') parameter and the stronger condition 
that the curvature\index{curvature} vanishes for all of its values. 
Bidifferential calculus\index{bidifferential calculus} \cite{DMH00a_DMH,DMH08bidiff_DMH} is the special case 
where $\nabla$ is linear in such a parameter. 
In particular, one meets this situation in case of the (anti-) self-dual Yang-Mills 
equation\index{Yang-Mills, self-dual} \cite{Maso+Wood96_DMH,Duna09_DMH}.
So let 
\begin{eqnarray*}
    \delta = \bar{\mathrm{d}} + \lambda \, \mathrm{d} \, , \qquad
    \boldsymbol{A} = A + \lambda \, B \, .
\end{eqnarray*}
The zero curvature\index{curvature} condition is required to hold for all $\lambda$ and is then equivalent to
\begin{eqnarray*}
    F_{\bar{\mathrm{d}}}[A] = 0 = F_{\mathrm{d}}[B] \, , \qquad 
    \mathrm{d} A + \bar{\mathrm{d}} B + A \, B + B \, A = 0 \, .
\end{eqnarray*}
Since typically one of the two ``gauge potentials'' $A$ and $B$ can be transformed to zero by a gauge transformation, 
we set $B=0$.\footnote{Setting alternatively $A=0$ and keeping $B$, corresponds to an exchange of 
$\mathrm{d}$ and $\bar{\mathrm{d}}$ in the following equations.}
Then the zero curvature\index{curvature} condition reduces to
\begin{eqnarray}
   \mathrm{d} A = 0  \, , \qquad F_{\bar{\mathrm{d}}}[A] = 0 \, .    \label{zero_curv_reduced_DMH}
\end{eqnarray}
The first equation can be solved by setting $A = \mathrm{d} \phi$, with $\phi \in \mathrm{Mat}(m,m,\mathcal{A})$, 
the algebra of $m \times m$ matrices over $\mathcal{A}$. Then the last equation takes the form
\begin{eqnarray}
     \mathrm{d} \, \bar{\mathrm{d}} \, \phi + \mathrm{d}\phi \; \mathrm{d}\phi = 0 \, .    \label{phi_eq_DMH}
\end{eqnarray}
Alternatively, we can solve $F_{\bar{\mathrm{d}}}[A] = 0$ by setting $A = (\bar{\mathrm{d}} g) \, g^{-1}$ with an invertible $g \in \mathrm{Mat}(m,m,\mathcal{A})$. 
Then $\mathrm{d} A = 0$ results in 
\begin{eqnarray}
        \mathrm{d} \, [ (\bar{\mathrm{d}} g) \, g^{-1} ] = 0  \, .   \label{g_eq_DMH}
\end{eqnarray}
Equations (\ref{phi_eq_DMH}) and (\ref{g_eq_DMH}) are related by 
\begin{eqnarray}
     \bar{\mathrm{d}} g  = (\mathrm{d} \phi) \, g \, ,  \label{Miura_DMH}
\end{eqnarray} 
which has both, (\ref{phi_eq_DMH}) and (\ref{g_eq_DMH}), as integrability\index{integrability} conditions. 

\begin{remark}
Another way to reduce the zero curvature condition to a single equation is by setting 
$A = ( \bar{\mathrm{d}} g - (\mathrm{d}g) \, \Delta ) \, g^{-1}$ with some $\Delta \in \mathcal{A}$, 
satisfying $\bar{\mathrm{d}} \Delta = (\mathrm{d} \Delta) \, \Delta$. Then (\ref{g_eq_DMH}) 
is generalized to $\mathrm{d} \, [ (\bar{\mathrm{d}} g - (\mathrm{d} g) \, \Delta ) \, g^{-1} ] = 0$.
\end{remark}

The problem of finding a Lax system\index{Lax system} for a PDDE\index{PDDE}, or a system of PDDEs, 
is now replaced by the problem of finding a 
bidifferential calculus\index{bidifferential calculus} representation, such that the PDDE (system) is equivalent 
to either (\ref{phi_eq_DMH}) or (\ref{g_eq_DMH}). These two equations can be regarded (cf. \cite{DMH08bidiff_DMH}) 
as generalizations or analogs of well-known potential forms of the famous (anti-) self-dual Yang-Mills 
equation\index{Yang-Mills, self-dual} on a flat four-dimensional space with Euclidean or split signature 
\cite{Maso+Wood96_DMH,Duna09_DMH}. 
We should point out that there may be PDDEs possessing a Lax system, but no bidifferential 
calculus\index{bidifferential calculus} representation. At least quite a number of integrable 
PDDEs\index{PDDE! integrable} are indeed realizations of one of the above equations. 

How to understand that an equation like (\ref{phi_eq_DMH}) has so many integrable equations as realizations? Of course, in order 
to make contact with an integrable PDDE\index{PDDE! integrable}, the algebra $\mathcal{A}$ has to include an algebra 
of functions on some set. But we can extend it by operators, so that $\phi$ becomes operator-valued. If, 
nevertheless, (\ref{phi_eq_DMH}) turns out to be equivalent to a PDDE, or a system of PDDEs\index{PDDE}, we are in business. 

In the bidifferential calculus\index{bidifferential calculus} framework, we thus consider ``universal'' equations 
for objects in any associative algebra\index{algebra! associative}. Since there are many concrete choices of a 
bidifferential calculus, there are lots of different realizations, 
which inherit essential properties from the universal framework. (\ref{phi_eq_DMH}) and (\ref{g_eq_DMH}) are important 
examples, but we will also show that they have interesting extensions. It is important to stress that 
bidifferential calculus\index{bidifferential calculus} is a \emph{framework}. It is not bound to the equations introduced above.  

Although examples presented in this work are integrable PDDEs\index{PDDE! integrable}, 
we want to emphasize the general structure, which only requires a graded associative algebra\index{algebra! graded} and 
two anticommuting graded derivations of degree one\index{derivation! of degree one} acting on it with coboundary properties. 
There might well be applications to quite different areas of mathematics. 

In Section~\ref{sec:bidiff_DMH} we define in a more precise way what we mean by a bi\-differential calculus. Some 
important examples are recalled. Furthermore, it is shown that (\ref{zero_curv_reduced_DMH}) is the 
integrability\index{integrability} condition of a linear system\index{linear system}. In this sense, (\ref{zero_curv_reduced_DMH}) is integrable. 
Moreover, there is a simple construction of an infinite chain of conservation laws\index{conservation law}.  

Section~\ref{sec:sym_DMH} introduces a class of symmetries of (\ref{zero_curv_reduced_DMH}). We show that they 
lead to a universal B\"acklund transformation\index{B\"acklund transformation} \cite{Roge+Shad82_DMH,Roge+Schi02_DMH} of (\ref{phi_eq_DMH}), respectively (\ref{g_eq_DMH}), 
as well as to (forward, backward and binary) Darboux transformations\index{Darboux transformation} (see \cite{Matv+Sall91_DMH} for concrete examples).  

Perhaps the most efficient method to generate large classes of exact solutions of integrable systems is a 
matrix version of the \emph{binary Darboux transformation} method. An abstract version in bidifferential calculus 
is recalled in Section~\ref{sec:bidiff_DMH}. Section~\ref{sec:recurrence_DMH} presents new material, an infinite system of equations 
solved by the latter method. 

In Section~\ref{sec:scs_DMH}, we sketch recent work about a deformation of the binary Darboux transformation 
in bidifferential calculus, leading to integrable equations with sources \cite{CDMH16_DMH,MHCT17_DMH}. 

Finally, Section~\ref{sec:conclusions_DMH} contains some concluding remarks.

\section{Bidifferential calculus}
\label{sec:bidiff_DMH}
First we give a precise definition of bidifferential calculus\index{bidifferential calculus}.

\begin{definition}
A \emph{graded associative algebra}\index{algebra! graded} is an associative algebra\index{algebra! associative} 
$\boldsymbol{\Omega} = \bigoplus_{r \geq 0} \boldsymbol{\Omega}^r$
over a field $\mathbb{K}$ of characteristic zero, where $\mathcal{A} := \boldsymbol{\Omega}^0$ is an associative 
algebra\index{algebra! associative} over $\mathbb{K}$ and $\boldsymbol{\Omega}^r$, $r \geq 1$,
are $\mathcal{A}$-bimodules\index{bimodule} such that $\boldsymbol{\Omega}^r \, \boldsymbol{\Omega}^s \subseteq \boldsymbol{\Omega}^{r+s}$.
A \emph{bidifferential calculus} is a unital graded associative algebra $\boldsymbol{\Omega}$, supplied
with two $\mathbb{K}$-linear graded derivations\index{derivation! graded} 
$\mathrm{d}, \bar{\mathrm{d}} : \boldsymbol{\Omega} \rightarrow \boldsymbol{\Omega}$
of degree one (hence $\mathrm{d} \boldsymbol{\Omega}^r \subseteq \boldsymbol{\Omega}^{r+1}$,
$\bar{\mathrm{d}} \boldsymbol{\Omega}^r \subseteq \boldsymbol{\Omega}^{r+1}$), and such that
\begin{eqnarray}
    \mathrm{d}^2 = \bar{\mathrm{d}}^2 = \mathrm{d} \bar{\mathrm{d}} + \bar{\mathrm{d}} \mathrm{d} = 0  \, .   
    \label{bidiff_conds_DMH}
\end{eqnarray}
\end{definition}

The next subsection explains the integrability\index{integrability} of the equations (\ref{phi_eq_DMH}) 
and (\ref{g_eq_DMH}), respectively their predecessor (\ref{zero_curv_reduced_DMH}), in 
bidifferential calculus\index{bidifferential calculus}. 
Another subsection recalls a construction of an infinite chain of conservation laws\index{conservation law}.  
The third subsection is devoted to integrable PDDEs\index{PDDE! integrable} as realizations of the equations (\ref{phi_eq_DMH}) or 
(\ref{g_eq_DMH}) in bidifferential calculus. 
Further subsections deal with bidifferential calculi associated with compatible Lie algebra structures, 
and with a topic in differential geometry\index{geometry! differential} known as Fr\"olicher-Nijenhuis 
theory\index{Fr\"olicher-Nijenhuis theory}.

\subsection{Linear systems and integrability}
Let $\Delta \in \mathrm{Mat}(n,n,\mathcal{A})$ and $\lambda$ be an $n \times n$ matrix of elements of $\boldsymbol{\Omega}^1$, 
subject to the following equations,
\begin{eqnarray}
      \bar{\mathrm{d}} \Delta + [\lambda , \Delta] = (\mathrm{d} \Delta) \, \Delta \, , \qquad  
      \bar{\mathrm{d}} \lambda + \lambda^2 = (\mathrm{d} \lambda) \, \Delta \, .   \label{Delta,lambda_eqs_DMH}
\end{eqnarray}
The linear system\index{linear system} 
\begin{eqnarray}
    \bar{\mathrm{d}} \theta = A \, \theta + (\mathrm{d} \theta) \, \Delta + \theta \, \lambda \, ,  \label{theta_eq_A_DMH}         
\end{eqnarray}
for $\theta \in \mathrm{Mat}(m,n,\mathcal{A})$, then has the integrability\index{integrability} condition
\begin{eqnarray*}
    F_{\bar{\mathrm{d}}}[A] \, \theta - (\mathrm{d} A) \, \theta \, \Delta = 0 \, .
\end{eqnarray*}
If $\Delta$ is such that this implies separate vanishing of both summands, (\ref{zero_curv_reduced_DMH}) turns out 
to be a consistency condition of (\ref{theta_eq_A_DMH}). (\ref{zero_curv_reduced_DMH}) is \emph{integrable} in the sense that it arises 
as the integrability\index{integrability} condition of a linear system\index{linear system}. 
The equations (\ref{Delta,lambda_eqs_DMH}) are in particular solved by $\lambda =0$ and a $\mathrm{d}$- and 
$\bar{\mathrm{d}}$-constant $\Delta$. But in some cases other solutions of (\ref{Delta,lambda_eqs_DMH}) are more important 
(see, e.g., \cite{DKMH11SIGMA_DMH,DMH13SIGMA_DMH}). 

In the same way one finds that (\ref{zero_curv_reduced_DMH}) is also the integrability\index{integrability} condition 
of the ``adjoint'' linear system\index{linear system}
\begin{eqnarray}
    \bar{\mathrm{d}} \eta = - \eta \, A + \Gamma \, \mathrm{d} \eta + \kappa \, \eta \, ,   \label{eta_eq_A_DMH}
\end{eqnarray}
where $\eta \in \mathrm{Mat}(n,m,\mathcal{A})$, $\Gamma \in \mathrm{Mat}(n,n,\mathcal{A})$, and $\kappa$ is an $n \times n$ matrix of
elements of $\boldsymbol{\Omega}^1$. They have to satisfy 
\begin{eqnarray}
      \bar{\mathrm{d}} \Gamma - [\kappa , \Gamma] = \Gamma \, \mathrm{d} \Gamma  \, , \qquad  
      \bar{\mathrm{d}} \kappa - \kappa^2 = \Gamma \, \mathrm{d} \kappa \, .   \label{Gamma,kappa_eqs_DMH}
\end{eqnarray}
(\ref{bidiff_conds_DMH}) and the graded derivation\index{derivation! graded} property of $\mathrm{d}$ and 
$\bar{\mathrm{d}}$ are sufficient to obtain these results.

\subsection{A recursive construction of conservation laws}
\label{subsec:conserv_DMH}
Let $(\boldsymbol{\Omega},\mathrm{d}, \bar{\mathrm{d}})$ be a bidifferential calculus\index{bidifferential calculus}. 
We call an element $J$ of $\boldsymbol{\Omega}^1$ a \emph{bi-closed 1-form}\index{bi-closed 1-form}\footnote{This 
generalizes the notion of a \emph{fundamental 1-form}\index{fundamental 1-form} introduced in \cite{Magr+Moro84_DMH} 
to bidifferential calculus. In the spirit of \cite{Grif+Mehd97_DMH}, it may also be called a ``conservation law''. 
Also see Section~\ref{subsec:FN_theory_DMH}.}
if 
\begin{eqnarray*}
      \mathrm{d} J = 0  \quad \mbox{and} \quad \bar{\mathrm{d}} J = 0 \, .
\end{eqnarray*}

Now let $J_1$ be a given bi-closed 1-form. If $\mathrm{d} J_1 = 0$ implies the existence of an element 
$\chi_1 \in \mathcal{A}$ such that
\begin{eqnarray*}
   J_1 = \mathrm{d} \chi_1 \, ,
\end{eqnarray*}
then it follows, by use of (\ref{bidiff_conds_DMH}), that 
\begin{eqnarray*}
    J_2 = \bar{\mathrm{d}} \chi_1 
\end{eqnarray*}
is also a bi-closed 1-form. This leads to a sequence of bi-closed 1-forms,
\begin{eqnarray*}
    J_{k+1} = \mathrm{d} \chi_{k+1} = \bar{\mathrm{d}} \chi_k \, , \qquad k=1,2,\ldots \, ,
\end{eqnarray*}
provided the closed $\mathrm{d}$-forms appearing in the construction are $\mathrm{d}$-exact 
(which is in particular true if the $\mathrm{d}$-cohomology is trivial).
Hence, if there is one bi-closed 1-form, there is typically an infinite number of bi-closed 1-forms.  

How to find a bi-closed 1-form to start with?
If $f \in \mathcal{A}$ satisfies $\bar{\mathrm{d}} \mathrm{d} f = 0$, then $J = \bar{\mathrm{d}} f$ is a bi-closed 1-form of 
the bidifferential calculus\index{bidifferential calculus}, since $\bar{\mathrm{d}} J = \bar{\mathrm{d}}^2 f = 0$ and  
$\mathrm{d} J = \mathrm{d} \bar{\mathrm{d}} f = - \bar{\mathrm{d}} \mathrm{d} f =0$. In particular, any non-zero 
$f \in \mathcal{A}$ satisfying $\mathrm{d} f = 0$ or $\bar{\mathrm{d}} f =0$ determines a bi-closed 1-form, which can 
then be used to apply the above construction of a chain of bi-closed 1-forms\index{bi-closed 1-form}.

\begin{remark}
If $J_1 = \bar{\mathrm{d}} \chi_0$ and $\bar{\mathrm{d}} \mathrm{d} \chi_0 = 0$, the formal power series
\begin{eqnarray*}
    \chi = \sum_{k=0}^\infty \chi_k \, \mu^k \, ,
\end{eqnarray*}
with a parameter or an indeterminate $\mu$, satisfies the linear system\index{linear system}
\begin{eqnarray*}
    \mathrm{d} \chi = \mu \, \bar{\mathrm{d}} \chi + \mathrm{d} \chi_0 \, .
\end{eqnarray*}
\end{remark}
\vspace{.2cm}

The above construction becomes really useful by observing that the zero curvature condition\index{curvature} 
(\ref{zero_curv_reduced_DMH}) is equivalent to the \emph{bicomplex conditions}
\begin{eqnarray*}
     \mathrm{d}^2 = 0 \, ,  \qquad \bar{D}^2 = 0 \, , \qquad \mathrm{d} \bar{D} + \bar{D} \mathrm{d} = 0 \, ,
\end{eqnarray*}
where we promoted $\mathrm{d}$ to the ``generalized covariant exterior derivative''\index{exterior derivative! generalized covariant} 
\begin{eqnarray*}
     \bar{D} = \bar{\mathrm{d}} - A \, . 
\end{eqnarray*}
Now $\mathrm{d}$ and $\bar{D}$ are acting on $\mathrm{Mat}(m,m,\Omega)$, regarded as a left $\mathcal{A}$-module\index{module}, and
$\chi_k \in \mathrm{Mat}(m,m,\mathcal{A})$.  
If the bicomplex conditions are equivalent to some PDDE\index{PDDE}, for a particular choice of the 
bidifferential calculus\index{bidifferential calculus}, 
then the above construction, with $\bar{\mathrm{d}}$ replaced by $\bar{D}$, yields an 
infinite chain of conservation laws\index{conservation law} for this PDDE. This has been elaborated for many integrable 
equations \cite{DMH00a_DMH}, also see \cite{DMH00c_DMH,DMH00e_DMH,DMH00ncKdV_DMH}. As a special case, it 
includes the construction of non-local conservation laws\index{conservation law} for principal chiral models\index{chiral model} 
in \cite{BIZZ79_DMH}, which has been formulated in terms of differential forms\index{differential form}  and generalized in 
\cite{DMH96int_DMH,DMH97int_DMH}.

\subsubsection{A modification}
\label{subsec:Magri_DMH}
Let $A = \mathrm{d} \phi$ and $\phi \in \mathrm{Mat}(m,m,\mathcal{A})$ be a solution of (\ref{phi_eq_DMH}). 
Then $\mathrm{d}$ and $\bar{D}$ satisfy the bicomplex conditions. 

We first note that (\ref{phi_eq_DMH}) can be expressed as
\begin{eqnarray*}
    \mathrm{d} (\bar{\mathrm{d}} \phi + \phi \, \mathrm{d} \phi ) = 0 \, .
\end{eqnarray*}
If this implies that the 1-form in the brackets is exact, there is $a_1 \in \mathrm{Mat}(m,m,\mathcal{A})$ such that
\begin{eqnarray}
     \mathrm{d} a_1 = \bar{\mathrm{d}} \phi + \phi \, \mathrm{d} \phi \, .   \label{da_1_DMH}
\end{eqnarray}
Acting with $\bar{\mathrm{d}}$ on this equation, using the latter 
to eliminate $\bar{\mathrm{d}} \phi$ and (\ref{phi_eq_DMH}) to replace the term $\bar{\mathrm{d}} \mathrm{d} \phi$, 
leads to
\begin{eqnarray*}
   0 = \bar{\mathrm{d}} \mathrm{d} a_1 - \mathrm{d} a_1 \, \mathrm{d}\phi 
     = - \mathrm{d}( \bar{\mathrm{d}} a_1 + a_1 \, \mathrm{d} \phi ) \, .
\end{eqnarray*}
We are thus led to
\begin{eqnarray*}
    \mathrm{d} a_2 = \bar{\mathrm{d}} a_1 + a_1 \, \mathrm{d} \phi \, .
\end{eqnarray*}
Continuing in this way and setting $a_0 = \phi$, we obtain the recurrence relation
\begin{eqnarray}
     \mathrm{d} a_{k+1} = \bar{\mathrm{d}} a_k + a_k \, \mathrm{d} a_0 \qquad  k=0,1,\ldots \, .  \label{a_recur_DMH}
\end{eqnarray}
Furthermore, we have
\begin{eqnarray*}
    \bar{\mathrm{d}} \mathrm{d} a_k = \mathrm{d} a_k \, \mathrm{d} \phi  \qquad k=0,1,\ldots \, .
\end{eqnarray*}

Let now $\tilde{\chi}_0 \in \mathrm{Mat}(m,m,\mathcal{A})$ be a solution of 
\begin{eqnarray*}
   0 = \mathrm{d} \bar{D} \tilde{\chi}_0 = \mathrm{d} ( \bar{\mathrm{d}} \tilde{\chi}_0 - (\mathrm{d} \phi) \, \tilde{\chi}_0 ) 
     = \mathrm{d} ( \bar{\mathrm{d}} \tilde{\chi}_0 + \phi \, \mathrm{d} \tilde{\chi}_0 ) \, .
\end{eqnarray*}
This means that\footnote{Here we depart from the previous scheme since $\tilde{J}_1 = J_1 + \mathrm{d} ( \phi \, \tilde{\chi}_0)$, 
where $J_1 = \bar{D} \tilde{\chi}_0$. $\tilde{J}_1$ is not annihilated by $\bar{D}$ if $\phi \neq 0$.} 
\begin{eqnarray*}
   \tilde{J}_1 = \bar{\mathrm{d}} \tilde{\chi}_0 + \phi \, \mathrm{d} \tilde{\chi}_0 
\end{eqnarray*}
is $\mathrm{d}$-closed. Assuming that it is $\mathrm{d}$-exact, there is 
$\tilde{\chi}_1 \in \mathrm{Mat}(m,m,\mathcal{A})$ such that
\begin{eqnarray*}
   \tilde{J}_1 = \mathrm{d} \tilde{\chi}_1 \, ,
\end{eqnarray*}
and it follows, by use of (\ref{da_1_DMH}), that
\begin{eqnarray*}
   \tilde{J}_2 = \bar{\mathrm{d}} \tilde{\chi}_1 + a_1 \, \mathrm{d} \tilde{\chi}_0  
\end{eqnarray*}
is $\mathrm{d}$-closed. Assuming again $\mathrm{d}$-exactness, we have 
\begin{eqnarray*}
 \tilde{J}_2 = \mathrm{d} \tilde{\chi}_2 \, . 
\end{eqnarray*}

Continuing in this way, by using (\ref{a_recur_DMH}) we obtain a chain of $\mathrm{d}$-closed 
(and then $\mathrm{d}$-exact) 1-forms
\begin{eqnarray*}
   \tilde{J}_{k+1} := \bar{\mathrm{d}} \tilde{\chi}_k + a_k \, \mathrm{d} \tilde{\chi}_0  
   = \mathrm{d} \tilde{\chi}_{k+1}      \qquad k=0,1,\ldots \, .
\end{eqnarray*}
If $a_k=0$, $k=1,2,\ldots$, this reduces to the previous construction in Section~\ref{subsec:conserv_DMH}. 
In the particular case of Fr\"olicher-Nijenhuis theory\index{Fr\"olicher-Nijenhuis theory}, 
also see Section~\ref{subsec:FN_theory_DMH}, 
this scheme reduces to that in \cite{Magr03_DMH} (also see \cite{Lore+Magri05_DMH}), where it is shown 
to appear even in the simple example of geodesic flow on $\mathbb{R}^n$. 

According to the main scheme in Section~\ref{subsec:conserv_DMH} we can also construct a chain of 
elements $\chi_k \in \mathrm{Mat}(m,m,\mathcal{A})$ satisfying  
\begin{eqnarray*}
       J_{k+1} := \bar{D} \chi_k = \mathrm{d} \chi_{k+1} \qquad  k=0,1,\ldots \, ,
\end{eqnarray*}
which only involves $a_0 = \phi$, but not $a_k$, $k>0$. Setting $\tilde{\chi}_0 = \chi_0$, we find 
\begin{eqnarray}
   \tilde{\chi}_k = \chi_k + \sum_{i+j=k-1 \atop i,j \geq 0} a_i \, \chi_j  \qquad  k=1,2,\ldots \, ,  \label{tchi-chi_rel_DMH}
\end{eqnarray}
modulo adding a $\mathrm{d}$-constant element of $\mathrm{Mat}(m,m,\mathcal{A})$. 

\begin{remark}
In \cite{Lore+Magri05_DMH} an infinite chain of constants of motion of integrable hierarchies\index{integrable hierarchy} 
of hydrodynamic type were constructed using Fr\"olicher-Nijenhuis theory\index{Fr\"olicher-Nijenhuis theory} 
(also see Section~\ref{subsec:FN_theory_DMH}). 
This is an example of our construction recalled in Section~\ref{subsec:conserv_DMH}, since the constants of motion  
are integrals of $\chi_k$, $k=0,1,\ldots$, which do not involve $a_k$, $k>0$. But the latter show up in the 
corresponding hierarchy of hydrodynamic type. We note that, in equation (2.24) of \cite{Lore+Magri05_DMH}, the inverse of 
the relation (\ref{tchi-chi_rel_DMH}) can be found, in different notation 
($a_i \mapsto -a_i$, $k_i = \chi_i$, $h_i = \tilde{\chi}_i$, $i=0,1,\ldots$).  
\end{remark}

All this motivates an exploration of the generalized scheme, presented above, in concrete examples. This will 
be left for future work.

\subsection{Integrable PDDEs obtained as realizations of equations in bidifferential calculus}
\label{subsec:PDDEs_in_bidiff_DMH} 
Given a unital associative algebra\index{algebra! associative} $\mathcal{A}$, a possible choice of $\boldsymbol{\Omega}$ is
\begin{eqnarray}
    \boldsymbol{\Omega} = \mathcal{A} \otimes \bigwedge \mathbb{C}^K \, ,    \label{Omega_wedge_DMH}
\end{eqnarray}
where $\bigwedge \mathbb{C}^K$ is the exterior algebra\index{algebra! exterior} of the vector space $\mathbb{C}^K$. In this case 
it is sufficient to define $\mathrm{d}$ and $\bar{\mathrm{d}}$ on $\mathcal{A}$, since they extend in an evident way 
to $\boldsymbol{\Omega}$, treating elements of $\bigwedge \mathbb{C}^K$ as $\mathrm{d}$- and $\bar{\mathrm{d}}$-constants. 
Let $\xi_1,\ldots,\xi_K$ be a basis of $\bigwedge^1 \mathbb{C}^K$. 
\vspace{.2cm}

\noindent
\textbf{Prototype example: self-dual Yang-Mills\index{Yang-Mills, self-dual}.} 
The (anti-) self-dual Yang-Mills equation on a flat four-dimensional space with Euclidean or split signature 
plays a prominent role as a four-dimensional integrable PDE system, from 
which many integrable PDEs can be derived via a reduction. It is also a prototype system in bidifferential 
calculus\index{bidifferential calculus}. Let $\mathcal{A}$ be the algebra of smooth complex functions of four 
(real or complex) variables $y, \bar{y}, z, \bar{z}$. Using (\ref{Omega_wedge_DMH}) with $K=2$, we define
\begin{eqnarray*}
    \mathrm{d} f = - f_y \, \xi_1 + f_z \, \xi_2 \, , \qquad
    \bar{\mathrm{d}} f = f_{\bar{z}} \, \xi_1 + f_{\bar{y}} \, \xi_2 \, ,
\end{eqnarray*}
on $\mathcal{A}$. Here a subscript indicates a partial derivative\index{partial derivative} with respect to the 
corresponding independent variable. 
With the above calculus, it follows that (\ref{phi_eq_DMH}) and (\ref{g_eq_DMH}) coincide with two well-known potential 
forms of the self-dual Yang-Mills  equation (cf. \cite{Maso+Wood96_DMH}), 
\begin{eqnarray*}
     \phi_{\bar{y} y} + \phi_{\bar{z} z} + [ \phi_y , \phi_z ] = 0  \, , \qquad
     (g_{\bar{y}} \, g^{-1})_y + (g_{\bar{z}} \, g^{-1})_z = 0 \, .
\end{eqnarray*}
Also see \cite{DMH08bidiff_DMH}, for example. 
\vspace{.2cm}

\noindent
\textbf{Further examples.}
Many integrable equations can be obtained as a reduction of the self-dual Yang-Mills equation\index{Yang-Mills, self-dual}. 
This holds in particular for the chiral model\index{chiral model} equation governing stationary and axially symmetric solutions 
of Einstein's equation\index{Einstein's equation} in electrovacuum, see \cite{DKMH11SIGMA_DMH,DMH13SIGMA_DMH} 
and references cited there. However, the Kadomtsev-Petviashvili (KP) equation\index{KP equation} apparently 
cannot be usefully obtained in this way. But it fits perfectly well into bidifferential calculus\index{bidifferential calculus}. 
Matrix versions of the most familiar integrable equations (including Korteweg-deVries\index{Korteweg-deVries equation} 
(KdV)\footnote{A bidifferential calculus for KdV is recalled 
in Example~\ref{ex:KdV_bidiff_DMH} in Section~\ref{sec:scs_DMH}.}, Boussinesq\index{Boussinesq equation}, 
Sine-Gordon\index{Sine-Gordon equation}, Nonlinear Schr\"odinger equation\index{Nonlinear Schr\"odinger equation} and its 
integrable discretizations, KP\index{KP equation}, Davey-Stewartson\index{Davey-Stewartson equation}, 
Toda\index{Toda lattice} and (generalized) Volterra lattices\index{Volterra lattice}, two-dimensional Toda lattice, 
discrete KP and Hirota bilinear difference equation\index{Hirota bilinear difference equation}) have been 
treated in the bidifferential calculus\index{bidifferential calculus} framework, and corresponding infinite families 
of exact (soliton) solutions have been generated. These examples are all based on a graded algebra\index{algebra! graded} 
of the form (\ref{Omega_wedge_DMH}). 
A comprehensive search for further examples still has to be carried out, with the aid of computer algebra.

\subsection{Compatible Lie algebra structures}
Let $\mathcal{V}$ be a vector space (over some field), supplied with a Lie algebra\index{algebra! Lie} structure, which 
is given by a skew-symmetric map $[ \, , \, ] : \mathcal{V} \otimes \mathcal{V} \to \mathcal{V}$, 
satisfying the Jacobi identity. 
Associated with it is the Chevalley-Eilenberg operator\index{Chevalley-Eilenberg operator} $\delta$, defined by
\begin{eqnarray*}
    (\delta \mathfrak{f})(v_1,v_2,\ldots,v_{k+1}) 
  &=& \sum_{i=1}^{k+1} (-1)^{i+1} [v_i , \mathfrak{f}(v_1,\ldots,\hat{v}_i,\ldots,v_{k+1})]  \\
  &&  + \sum_{i<j} (-1)^{i+j} \mathfrak{f}([v_i,v_j],v_1,\ldots,\hat{v}_i,\ldots,\hat{v}_j,\ldots,v_{k+1}) \, ,
\end{eqnarray*}
acting on alternating multi-linear maps $\mathfrak{f} : \bigwedge^k \mathcal{V} \rightarrow \mathcal{V}$. 
Here a hat indicates an omission.
Given two compatible Lie algebra\index{algebra! Lie} structures (which means that any linear combination is also a 
Lie algebra structure),   
the corresponding Chevalley-Eilenberg operators $\delta_1$ and $\delta_2$ on 
$\boldsymbol{\Omega} = \mathcal{V} \oplus \bigoplus_{k=1}^\infty \mathrm{Hom}(\bigwedge^k \mathcal{V},\mathcal{V})$ 
constitute a bidifferential calculus. 
The case of two compatible Poisson structures\index{Poisson structure} on $\mathcal{V} = C^\infty(M)$, on 
a manifold\index{manifold} $M$, makes contact with the framework 
of bi-Hamiltonian systems\index{bi-Hamiltonian system}, also see the next subsection.\footnote{F M-H had a 
very illuminating discussion about all this with Martin Bordemann in June 2000.}

\subsection{Fr\"olicher-Nijenhuis theory}
\label{subsec:FN_theory_DMH}
A \emph{Nijenhuis tensor}\index{Nijenhuis tensor} on a manifold\index{manifold} $M$ is a tensor field $N$ of type (1,1) (or a fiber 
preserving endomorphism of $TM$) on $M$, with vanishing Nijenhuis torsion\index{Nijenhuis torsion}, i.e.,  
\begin{eqnarray*}
    T(N)(X,Y) := [NX,NY] - N( [NX,Y] + [X,NY] - N [X,Y] ) = 0 \, ,
\end{eqnarray*}
where $X,Y$ are vector fields on $M$, and $[ \, , \, ]$ denotes the usual commutator on 
the space of vector fields. In the context of integrable systems, we should mention the notion of a 
\emph{Poisson-Nijenhuis structure}\index{Poisson-Nijenhuis structure} on $M$, which is given by a Poisson 
structure\index{Poisson structure} and a Nijenhuis tensor, satisfying a certain compatibility condition 
\cite{Magr+Moro84_DMH,Kosm+Magr90_DMH}. 
This structure has been shown to be an integrability\index{integrability} feature of many integrable Hamiltonian systems. 

Given a Nijenhuis tensor, a bidifferential calculus\index{bidifferential calculus} is obtained as follows. 
Let $\boldsymbol{\Omega}$ be the algebra of differential forms\index{algebra! of differential forms} on $M$. 
Then the exterior derivative\index{exterior derivative} $\mathrm{d}$ 
and $\bar{\mathrm{d}} = \mathrm{d}_N := i_N \, \mathrm{d}$ satisfy (\ref{bidiff_conds_DMH}). 
Here we look at $N$ as a vector-valued 1-form, and $i_N$ acts via contraction of a vector field and a 1-form.  
According to Fr\"olicher-Nijenhuis theory\index{Fr\"olicher-Nijenhuis theory} \cite{Froe+Nije56_DMH}, any 
derivation $\bar{\mathrm{d}}$ of degree one\index{derivation! of degree one}, which 
anti-commutes with $\mathrm{d}$ and satisfies $\bar{\mathrm{d}}^2=0$, is of the form $\mathrm{d}_N$. 

In \cite{Magr03_DMH} a \emph{recursion operator}\index{recursion operator} on a symplectic 
manifold\index{manifold! symplectic} $(M,\omega)$ 
is defined to be a Nijenhuis tensor that is compatible with the symplectic 2-form $\omega$. This means that 
$\omega(NX,Y)$ defines a 2-form and $\mathrm{d}_N \omega =0$. This reformulates a crucial structure of 
bi-Hamiltonian systems\index{bi-Hamiltonian system} (a subclass of integrable systems), which serves to construct a so-called 
\emph{Lenard chain}\index{Lenard chain} of conserved quantities in involution. We also refer to 
\cite{DMH00a_DMH,CST00a_DMH,Cram+Sarl02_DMH,Chav03_DMH,Chav05_DMH,Tondo06_DMH,Lore+Magri05_DMH,Lorenzoni06_DMH,Cama+More10,Arsi+Lore13_DMH} 
for related work. 

Using the (local $\mathrm{d}$-) exactness of ($\mathrm{d}$-) closed 1-forms, in \cite{Magr03_DMH} (also see \cite{Lore+Magri05_DMH}) 
the construction of a (generalized) Lenard chain\index{Lenard chain} is formulated concisely in terms 
of $\mathrm{d}$ and $\mathrm{d}_N$. 
This is a special case of the recursive construction of $\mathrm{d}$-closed 1-forms in 
bidifferential calculus\index{bidifferential calculus}, presented in Section~\ref{subsec:Magri_DMH}. 
In the present context, a bi-closed 1-form $J$ (cf. Section~\ref{subsec:conserv_DMH}) is defined by 
$\mathrm{d} J = 0 = \mathrm{d}_N J$ and called a \emph{fundamental 1-form}\index{fundamental 1-form} \cite{Magr+Moro84_DMH}. 
In \cite{Grif+Mehd97_DMH} it has also been called a conservation law\index{conservation law} for a type (1,1) tensor field $N$. 

\begin{remark}
Via $(N^\ast \alpha)(X) := \alpha(NX)$, $N$ determines an endomorphism $N^\ast$ of the space of 1-forms.
The (generalized) Hodge star operator\index{Hodge star operator} used in \cite{DMH96int_DMH,DMH97int_DMH} is a 
linear map on the space of 1-forms and thus plays the role of $N^\ast$.
\end{remark}

\section{Symmetries in bidifferential calculus}
\label{sec:sym_DMH}
Let us consider the system (\ref{zero_curv_reduced_DMH}). A \emph{symmetry}\index{symmetry} of it should be any transformation 
of $A$ that leaves the system invariant, after implementing the condition that $A$ is a solution. 
The first of equations (\ref{zero_curv_reduced_DMH}) is obviously invariant under
\begin{eqnarray*}
     A \mapsto A' = A + \mathrm{d} \sigma \, ,
\end{eqnarray*}
with a 0-form (element of $\boldsymbol{\Omega}^0$) $\sigma$. The symmetry condition obtained 
from the second of equations (\ref{zero_curv_reduced_DMH}) is then
\begin{eqnarray*}
    \mathrm{d} \Big( \bar{\mathrm{d}} \sigma - [ A , \sigma ] + \sigma \, \mathrm{d} \sigma \Big) = 0 \, ,
\end{eqnarray*}
where we used that $A$ shall be a solution of (\ref{zero_curv_reduced_DMH}). If we assume that the $\mathrm{d}$-closed 
1-form in the brackets is $\mathrm{d}$-exact, there is a ``potential'' $\rho$ such that
\begin{eqnarray*}
    \bar{\mathrm{d}} \sigma - [ A , \sigma ] + \sigma \, \mathrm{d} \sigma  = \mathrm{d} \rho \, .
\end{eqnarray*}
In the following we show that B\"acklund transformations\index{B\"acklund transformation} and 
Darboux transformations\index{Darboux transformation} are special 
classes of such symmetries. B\"acklund transformations arise via elimination of $\sigma$ from the symmetry\index{symmetry} conditions. 
In case of Darboux transformations, a solution $\sigma$ is constructed by use of the linear or adjoint linear system\index{linear system}.

\subsection{B\"acklund transformation} 
We show how B\"acklund transformations\index{B\"acklund transformation} emerge from the above class of symmetries.
If we set $\rho=0$ and use $\mathrm{d} \sigma = A' - A$, we obtain
\begin{eqnarray*}
    \bar{\mathrm{d}} \sigma =  A \, \sigma - \sigma \, A' \, .
\end{eqnarray*}
There are now the following two cases (also see \cite{DMH01bt_DMH,DMH08bidiff_DMH}). 
\vspace{.2cm}

\noindent
\textbf{1.} Writing $A = \mathrm{d} \phi$ and $A' = \mathrm{d} \phi'$, the equation $\mathrm{d} \sigma = A' - A$
is solved by 
\begin{eqnarray*}
    \sigma = \phi' - \phi + C \, ,
\end{eqnarray*}
where $\mathrm{d} C =0$. The remaining symmetry\index{symmetry} condition then reads
\begin{eqnarray*}
    \bar{\mathrm{d}} (\phi' - \phi + C) =  \mathrm{d} \phi \, (\phi' - \phi + C) - (\phi' - \phi + C) \, \mathrm{d} \phi' 
    \, , \qquad \mathrm{d} C =0 \, .
\end{eqnarray*}
This is a relation between two solutions of (\ref{phi_eq_DMH}) and thus constitutes a B\"acklund 
transformation for the equation (\ref{phi_eq_DMH}). 
\vspace{.2cm}

\noindent
\textbf{2.} Alternatively, using $A = (\bar{\mathrm{d}} g) \, g^{-1}$ and $A' = (\bar{\mathrm{d}} g') \, {g'}^{-1}$, we are led 
to
\begin{eqnarray*}
    \sigma = g \, K \, {g'}^{-1} \, , \qquad \bar{\mathrm{d}} K =0 \, ,
\end{eqnarray*}
and thus
\begin{eqnarray*}
    (\bar{\mathrm{d}} g') \, {g'}^{-1} -  (\bar{\mathrm{d}} g) \, g^{-1} = \mathrm{d} (g \, K \, {g'}^{-1}) \, , \qquad \bar{\mathrm{d}} K =0 \, .
\end{eqnarray*}
This is a B\"acklund transformation for the equation (\ref{g_eq_DMH}). 
\vspace{.2cm}

\noindent
In \cite{DMH01bt_DMH} we derived and recovered B\"acklund transformations from the above general 
expressions for quite a number of integrable equations. 

\subsubsection{Permutability}
Let $A_i, A_k$ be two solutions of (\ref{zero_curv_reduced_DMH}), related by a B\"acklund transformation, hence
\begin{eqnarray*}
   \mathrm{d} \sigma_{i,k} = A_k - A_i \, , \qquad
   \bar{\mathrm{d}} \sigma_{i,k} =  A_i \, \sigma_{i,k} - \sigma_{i,k} \, A_k \, .  
\end{eqnarray*}
This can be expressed as a refactorization condition,
\begin{eqnarray*}
     ( \bar{\mathrm{d}} - A_i - \lambda^{-1} \mathrm{d} ) \, \mathcal{B}(\lambda)_{i,k}
     = \mathcal{B}(\lambda)_{i,k} \, ( \bar{\mathrm{d}} - A_k - \lambda^{-1} \mathrm{d} )  \qquad \forall \lambda \, ,
\end{eqnarray*}
acting from the left on $\mathrm{Mat}(m,m,\Omega)$. Here we set
\begin{eqnarray*}
     \mathcal{B}(\lambda)_{i,k} = 1 + \lambda \, \sigma_{i,k} \, .
\end{eqnarray*}
Suppose there are two different chains of B\"acklund transformations, 
\begin{eqnarray*}
  A_i \stackrel{ \mathcal{B}(\lambda)_{i,k} }{\longmapsto} A_k \stackrel{ \mathcal{B}(\lambda)_{k,j} }{\longmapsto} A_j 
      \, , \qquad
  A_i \stackrel{ \mathcal{B}(\lambda)_{i,l} }{\longmapsto} A_{l} \stackrel{ \mathcal{B}(\lambda)_{l,j} }{\longmapsto} A_j  
   \, ,
\end{eqnarray*}
starting with the same solution and ending with the same solution. Then the refactorization equation implies
\begin{eqnarray*}
     \mathcal{B}(\lambda)_{i,k} \, \mathcal{B}(\lambda)_{k,j} = \mathcal{B}(\lambda)_{i,l} \, \mathcal{B}(\lambda)_{l,j} 
      \qquad \forall \lambda \, ,
\end{eqnarray*}
which requires the \emph{permutability conditions}
\begin{eqnarray*}
     \sigma_{i,k} + \sigma_{k,j} = \sigma_{i,l} + \sigma_{l,j} \, , \qquad
     \sigma_{i,k} \, \sigma_{k,j} = \sigma_{i,l} \, \sigma_{l,j} \, .
\end{eqnarray*}
Referring to the two cases above, this becomes
\begin{eqnarray*}
   (\phi_k - \phi_i + C_{i,k})(\phi_j - \phi_k + C_{k,j}) &=& (\phi_l - \phi_i + C_{i,l}) (\phi_j - \phi_l + C_{l,j}) \, , \\
   C_{i,k} + C_{k,j} &=& C_{i,l} + C_{l,j} \, ,
\end{eqnarray*}
respectively
\begin{eqnarray*}
   g_i K_{i,k} \, g_k^{-1} + g_k K_{k,j} \, g_j^{-1} &=& g_i K_{i,l} \, g_l^{-1} + g_l K_{l,j} \, g_j^{-1} \, , \\
   K_{i,k} \, K_{k,j} &=& K_{i,l} \, K_{l,j} \, .
\end{eqnarray*}
For concrete choices of bidifferential calculi\index{bidifferential calculus} associated with integrable 
PDDEs\index{PDDE! integrable}, corresponding ``permutability theorems''\index{permutability theorem} 
\cite{Roge+Shad82_DMH,Roge+Schi02_DMH} arise from these equations. If two solutions are obtained by (elementary) 
B\"acklund transformations (involving different parameters) applied to a given solution, then a fourth solution is 
determined in a purely algebraic way. 
See \cite{DMH01bt_DMH} for many examples, where permutability theorems are recovered via the 
bidifferential calculus\index{bidifferential calculus} approach.

\subsection{Forward Darboux transformation} 
For given solutions
$\Delta_\theta, \lambda_\theta$ of (\ref{Delta,lambda_eqs_DMH}), with $n=m$, and a solution $A_0$ of 
(\ref{zero_curv_reduced_DMH}), let $\theta \in \mathrm{Mat}(m,m,\mathcal{A})$ be an invertible solution of the 
linear system\index{linear system} 
\begin{eqnarray*}
    \bar{\mathrm{d}} \theta = A_0 \, \theta + (\mathrm{d} \theta) \, \Delta_\theta + \theta \, \lambda_\theta \, ,
\end{eqnarray*}
which is (\ref{theta_eq_A_DMH}) with $n=m$ and $A$ replaced by $A_0$. Here solutions of (\ref{Delta,lambda_eqs_DMH}) 
carry a subscript $\theta$ just to remind us of the fact that these are the solutions of (\ref{Delta,lambda_eqs_DMH})
that appear as coefficients in the linear equation for $\theta$. This convention will be used further on in the present 
section, but not beyond it. 

Setting 
\begin{eqnarray*}
     \sigma_{[+1]} = \theta \, \Delta_\theta \, \theta^{-1} \, ,
\end{eqnarray*}
it follows that
\begin{eqnarray*}
     A_{[+1]} = A_0 + \mathrm{d} \sigma_{[+1]} 
\end{eqnarray*}
also satisfies (\ref{zero_curv_reduced_DMH}). 
Let $\psi$ be a solution of the linear system\index{linear system}
\begin{eqnarray*}
    \bar{\mathrm{d}} \psi = A_0 \, \psi + (\mathrm{d} \psi) \, \Delta_\psi + \psi \, \lambda_\psi \, , 
\end{eqnarray*}
where $\Delta_\psi,\lambda_\psi$ satisfy (\ref{Delta,lambda_eqs_DMH}). Then
\begin{eqnarray*}
    \psi_{[+1]} = \sigma_{[+1]} \psi - \psi \, \Delta_\psi 
\end{eqnarray*}
satisfies
\begin{eqnarray*}
   \bar{\mathrm{d}} \psi_{[+1]} = A_{[+1]} \, \psi_{[+1]} + (\mathrm{d} \psi_{[+1]}) \, \Delta_\psi + \psi_{[+1]} \, \lambda_\psi \, .
\end{eqnarray*}
 
\subsection{Backward Darboux transformation} 
For given solutions $\Gamma_\eta, \kappa_\eta$ of (\ref{Gamma,kappa_eqs_DMH}) with $n=m$, and a solution $A_0$ of 
(\ref{zero_curv_reduced_DMH}), let $\eta \in \mathrm{Mat}(m,m,\mathcal{A})$ be an invertible solution of the 
adjoint linear system\index{linear system} 
\begin{eqnarray*}
   \bar{\mathrm{d}} \eta = - \eta \, A_0 + \Gamma_\eta \, \mathrm{d} \eta + \kappa_\eta \, \eta \, , 
\end{eqnarray*}
which is (\ref{eta_eq_A_DMH}) with $n=m$, and $A_0$ instead of $A$. Setting 
\begin{eqnarray*}
     \sigma_{[-1]} = -\eta^{-1} \, \Gamma_\eta \, \eta \, ,
\end{eqnarray*}
it follows that
\begin{eqnarray*}
     A_{[-1]} = A_0 + \mathrm{d} \sigma_{[-1]} 
\end{eqnarray*}
also satisfies (\ref{zero_curv_reduced_DMH}). 
Let $\chi$ be a solution of the linear system\index{linear system}
\begin{eqnarray*}
     \bar{\mathrm{d}} \chi = - \chi \, A_0 + \Gamma_\chi \, \mathrm{d} \chi + \kappa_\chi \, \chi \, , 
\end{eqnarray*}
where $\Gamma_\chi,\kappa_\chi$ satisfy (\ref{Gamma,kappa_eqs_DMH}). Then
\begin{eqnarray*}
    \chi_{[-1]} = \Gamma_\chi \chi + \chi \, \sigma_{[-1]}  
\end{eqnarray*}
satisfies
\begin{eqnarray*}
   \bar{\mathrm{d}} \chi_{[-1]} = - \chi_{[-1]} A_{[-1]} + \Gamma_\chi \, \mathrm{d} \chi_{[-1]} + \kappa_\chi \, \chi_{[-1]}  \, .
\end{eqnarray*}

\subsection{Binary Darboux transformation}
\label{subsec:bDt_DMH}
The following essentially generalizes the binary Darboux transformation method as formulated, e.g., in \cite{Matv+Sall91_DMH},  
to bidifferential calculus\index{bidifferential calculus}. 
In order to combine forward and backward Darboux transformations, we have to look for a transformation of 
$\chi$ under the forward Darboux transformation and a transformation of $\psi$ under the backward Darboux transformation. 
Let $\Omega(\chi,\psi)$ satisfy the consistent linear equations
\begin{eqnarray*}
  \Gamma_\chi \, \Omega(\chi,\psi) - \Omega(\chi,\psi) \, \Delta_\psi &=& \chi \, \psi \, ,   \\
  \bar{\mathrm{d}} \Omega(\chi,\psi) &=& (\mathrm{d} \Omega(\chi,\psi)) \, \Delta_\psi - (\mathrm{d} \Gamma_\chi) \, \Omega(\chi,\psi) \\
  &&   + (\mathrm{d} \chi) \, \psi + \kappa_\chi \, \Omega(\chi,\psi) + \Omega(\chi,\psi) \, \lambda_\psi  \, .  
\end{eqnarray*}
Then
\begin{eqnarray*}
   \psi_{[-1]} := \eta^{-1} \Omega(\eta,\psi) \, , \qquad
   \chi_{[+1]} := \Omega(\chi,\theta) \, \theta^{-1} \, , 
\end{eqnarray*}
indeed satisfy the linear equations
\begin{eqnarray*}
  && \bar{\mathrm{d}} \psi_{[-1]} = A_{[-1]} \, \psi_{[-1]} + (\mathrm{d} \psi_{[-1]}) \, \Delta_\psi + \psi_{[-1]} \, \lambda_\psi \, , \\
  && \bar{\mathrm{d}} \chi_{[+1]} = - \chi_{[+1]} \, A _{[+1]} + \Gamma_\chi \, \mathrm{d} \chi_{[+1]} + \kappa_\chi \, \chi_{[+1]} \, .
\end{eqnarray*}
Inserting
\begin{eqnarray*}
     \theta_{[-1]} = \eta^{-1} \Omega(\eta,\theta) \, , \qquad
     \eta_{[+1]} = \Omega(\eta,\theta) \, \theta^{-1}       \, ,
\end{eqnarray*}
in the expressions for $\sigma_{[\pm 1]}$, yields
\begin{eqnarray*}
  &&  \sigma_{[+1,-1]} := - \eta_{[+1]}^{-1} \, \Gamma_\eta \, \eta_{[+1]} 
                     = - \theta \, \Omega(\eta,\theta)^{-1} \Gamma_\eta \, \Omega(\eta,\theta) \, \theta^{-1}   \, , \\
  &&  \sigma_{[-1,+1]} := \theta_{[-1]} \, \Delta_\theta \, \theta_{[-1]}^{-1} 
                       = \eta^{-1} \Omega(\eta,\theta) \Delta_\theta \, \Omega(\eta,\theta)^{-1} \eta \, ,
\end{eqnarray*}
and we obtain
\begin{eqnarray*}
    A_{[+1,-1]} &:=& A_{[+1]} + \mathrm{d} \sigma_{[+1,-1]} = A_0 - \mathrm{d} (\theta \, \Omega(\eta,\theta)^{-1} \eta) \\ 
      &=& A_{[-1]} + \mathrm{d} \sigma_{[-1,+1]} ) =: A_{[-1,+1]} \, .
\end{eqnarray*}
Furthermore, 
\begin{eqnarray*}
  &&  \chi_{[+1,-1]} = \Gamma_\chi \, \chi_{[+1]} + \chi_{[+1]} \, \sigma_{[+1,-1]}
                    = \chi - \Omega(\chi,\theta) \, \Omega(\eta,\theta)^{-1} \eta \, , \\
  &&  \psi_{[-1,+1]} = \sigma_{[-1,+1]} \, \psi_{[-1]} - \psi_{[-1]} \Delta_\theta 
                      = \psi - \theta \, \Omega(\eta,\theta)^{-1} \Omega(\eta,\psi) \, ,
\end{eqnarray*}
satisfy the linear equations
\begin{eqnarray*}
  && \bar{\mathrm{d}} \chi_{[+1,-1]} = - \chi_{[+1,-1]} \, A_{[+1,-1]} + \Gamma_\chi \, \mathrm{d} \chi_{[+1,-1]} + \kappa_\chi \, \chi_{[+1,-1]} \, , \\
  && \bar{\mathrm{d}} \psi_{[-1,+1]} = A_{[-1,+1]} \, \psi_{[-1,+1]} + (\mathrm{d} \psi_{[-1,+1]}) \, \Delta_\psi
                          + \psi_{[-1,+1]} \, \lambda_\psi   \, .
\end{eqnarray*}
The transformation
\begin{eqnarray*}
    (A_0,\psi,\chi) \longmapsto (A_{[+1,-1]},\psi_{[-1,+1]},\chi_{[+1,-1]})
\end{eqnarray*}
is the \emph{binary Darboux transformation}.
\vspace{.2cm}

\noindent
\textbf{1.} Using $A = \mathrm{d} \phi$, we see that 
\begin{eqnarray*}
    \phi_{[+1,-1]} = \phi_{[-1,+1]} = \phi_0 - \theta \, \Omega(\eta,\theta)^{-1} \eta 
\end{eqnarray*}
solves (\ref{phi_eq_DMH}).
\vspace{.2cm}

\noindent
\textbf{2.} Using instead $A = (\bar{\mathrm{d}} g) \, g^{-1}$, we conclude that
\begin{eqnarray*}
    g_{[-1,+1]} = g_{[+1,-1]} = (1 - \theta \, \Omega(\eta,\theta)^{-1} \Gamma_\eta^{-1} \eta) \, g_0 
\end{eqnarray*}
solves (\ref{g_eq_DMH}). We note that
\begin{eqnarray*} 
  g_{[-1,+1]}^{-1} = g_0^{-1} (1 + \theta \, \Delta_\theta^{-1} \Omega(\eta,\theta)^{-1} \eta) \, .    
\end{eqnarray*}

\begin{remark}
Setting
\begin{eqnarray*}
    q = \eta_{[+1]}^{-1} \, , \qquad
    r = \theta_{[-1]}^{-1} \, ,
\end{eqnarray*}
we find that they satisfy the linear system\index{linear system}
\begin{eqnarray*}
   && \bar{\mathrm{d}} q = A_{[+1,-1]} \, q + (\mathrm{d} q) \, \Gamma_\eta + q \, (\mathrm{d} \Gamma_\eta - \kappa_\eta) \, , \\
   && \bar{\mathrm{d}} r = - r \, A_{[-1,+1]} + \Delta_\theta \, \mathrm{d} r + (\mathrm{d} \Delta_\theta - \lambda_\theta) \, r \, ,
\end{eqnarray*}
 from which we read off that $\Delta_q = \Gamma_\eta$, $\lambda_q = \mathrm{d} \Gamma_\eta - \kappa_\eta$,  
$\Gamma_r = \Delta_\theta$ and $\kappa_r = \mathrm{d} \Delta_\theta - \lambda_\theta$. 
Furthermore, we have
\begin{eqnarray*}
    \Delta_\theta \, \Omega(r,q) - \Omega(r,q) \, \Gamma_\eta = r \, q \, , 
\end{eqnarray*}
and
\begin{eqnarray*}
   \bar{\mathrm{d}} \Omega(r,q) &=& (\mathrm{d} \Omega(r,q)) \, \Gamma_\eta - (\mathrm{d} \Delta_\theta) \, \Omega(r,q) 
     + (\mathrm{d} r) \, q \\
     && + (\mathrm{d} \Delta_\theta - \lambda_\theta) \, \Omega(r,q) + \Omega(r,q) \, (\mathrm{d} \Gamma_\eta - \kappa_\eta) \\
     &=& (\mathrm{d} \Omega(r,q) \, \Gamma_\eta) + (\mathrm{d} r) \, q - \Omega(r,q) \, \kappa_\eta - \lambda_\theta \, \Omega(r,q) \, .
\end{eqnarray*}
Inserting our expressions for $q$ and $r$, we can conclude that these equations are solved by setting
\begin{eqnarray*}
    \Omega(r,q) = - \Omega(\eta,\theta)^{-1} \, .
\end{eqnarray*}
Moreover, one can show that 
\begin{eqnarray*}
    \Omega(r,\psi_{[-1,+1]}) = \Omega(\eta,\theta)^{-1} \Omega(\eta,\psi) \, , \qquad
    \Omega(\chi_{[+1,-1]},q) = \Omega(\chi,\theta) \, \Omega(\eta,\theta)^{-1}  \, .
\end{eqnarray*}
Using $q$ and $r$ in the binary Darboux transformation, we find
\begin{eqnarray*}
  && A_{[+1,-1]} -\mathrm{d}( q \, \Omega(r,q)^{-1} r) = A_0 \, , \\
  && \psi_{[-1,+1]} - q \, \Omega(r,q)^{-1} \Omega(r,\psi_{[-1,+1]}) = \psi  \, , \\
  &&  \chi_{[-1,+1]} - \Omega(\chi_{[-1,+1]},q) \, \Omega(r,q)^{-1} = \chi  \, . 
\end{eqnarray*}
Hence this leads back to what we started with.
\end{remark}

 From the above Darboux transformations\index{Darboux transformation} one recovers known Darboux transformations 
\cite{Matv+Sall91_DMH} for many integrable systems via corresponding choices of the 
bidifferential calculus\index{bidifferential calculus}.

\section{A matrix version of the binary Darboux transformation}
\label{sec:bDt_DMH}
Now we drop the restriction $n=m$, imposed in the preceding section. Let us recall that 
(\ref{phi_eq_DMH}) is \emph{integrable} in the sense that it arises as the integrability\index{integrability} 
condition of the linear system\index{linear system} 
\begin{eqnarray}
    \bar{\mathrm{d}} \theta = (\mathrm{d} \phi) \, \theta + (\mathrm{d} \theta) \, \Delta + \theta \, \lambda \, .  \label{theta_eq_DMH}         
\end{eqnarray}
This is (\ref{theta_eq_A_DMH}) with $A = \mathrm{d}\phi$. Of course, $\Delta$ and $\lambda$ have to satisfy 
(\ref{Delta,lambda_eqs_DMH}). 

(\ref{phi_eq_DMH}) is also the integrability\index{integrability} condition of the ``adjoint'' linear system\index{linear system}
\begin{eqnarray}
    \bar{\mathrm{d}} \eta = - \eta \, \mathrm{d} \phi + \Gamma \, \mathrm{d} \eta + \kappa \, \eta \, ,   \label{eta_eq_DMH}
\end{eqnarray}
where $\Gamma$ and $\kappa$ have to solve (\ref{Gamma,kappa_eqs_DMH}). 

Let $\Omega$ be a solution of the linear equations
\begin{eqnarray}
  &&  \Gamma \, \Omega - \Omega \, \Delta = \eta \, \theta \, ,   \label{Omega_Sylvester_DMH} \\
  &&  \bar{\mathrm{d}} \Omega = (\mathrm{d} \Omega) \, \Delta - (\mathrm{d} \Gamma) \, \Omega 
     + \kappa \, \Omega + \Omega \, \lambda + (\mathrm{d} \eta) \, \theta \, .  \label{bd_Omega_DMH} 
\end{eqnarray}
The equation obtained by acting with $\bar{\mathrm{d}}$ on (\ref{Omega_Sylvester_DMH}) 
is identically satisfied as a consequence of (\ref{theta_eq_DMH}), (\ref{eta_eq_DMH}), (\ref{Omega_Sylvester_DMH}), 
and the equation that results from (\ref{Omega_Sylvester_DMH}) by acting with $\mathrm{d}$ on it.   
Correspondingly, also the equation that results from acting with $\bar{\mathrm{d}}$ on (\ref{bd_Omega_DMH}) is identically satisfied 
as a consequence of the preceding equations. If $\phi_0$ is a given solution of (\ref{phi_eq_DMH}), it follows 
\cite{DMH08bidiff_DMH,DMH13SIGMA_DMH} that
\begin{eqnarray*}
   \phi = \phi_0 - \theta \, \Omega^{-1} \, \eta 
\end{eqnarray*}
is a new solution of (\ref{phi_eq_DMH}), and 
\begin{eqnarray*}  
   q = \theta \, \Omega^{-1} \, , \qquad 
   r = \Omega^{-1} \, \eta 
\end{eqnarray*}
satisfy 
\begin{eqnarray*}
    \bar{\mathrm{d}} q = (\mathrm{d} \phi) \, q + \mathrm{d}( q \, \Gamma) - q \, \kappa  \, ,  \qquad
    \bar{\mathrm{d}} r = - r \, \mathrm{d} \phi + \mathrm{d}( \Delta \, r) - \lambda \, r  \, .       
\end{eqnarray*}

There is an analogous solution-generating result for equation (\ref{g_eq_DMH}), see \cite{CDMH16_DMH} and references 
cited there.

\section{An infinite system of equations in bidifferential calculus with solutions generated 
        by the binary Darboux transformation}
\label{sec:recurrence_DMH}        
Are there perhaps other equations in bidifferential calculus\index{bidifferential calculus} that can also be solved via the binary 
Darboux transformation\index{Darboux transformation}? 
Let 
\begin{eqnarray*}
    \Phi_{i,j} = \theta \, \Delta^i \Omega^{-1} \Gamma^j \, \eta   \qquad \quad i,j \in \mathbb{Z} \, .
\end{eqnarray*}
A corresponding generating function is
\begin{eqnarray*}
     \Phi(s,t) = \theta \, (1-s \, \Delta)^{-1} \Omega^{-1} (1-t \, \Gamma)^{-1} \eta \, ,
\end{eqnarray*}
which has the following formal power series expansions, 
\begin{eqnarray*}
   && \Phi(s,t) = \sum_{i,j=0}^\infty \Phi_{i,j} \, s^i t^j \qquad \mbox{around } (s,t)=(0,0) \, , \\
   && \Phi(s,t) = - \sum_{i=1,j=0}^\infty \Phi_{-i,j} \, s^{-i} t^j \qquad \mbox{around } (s,t)=(\infty,0) \, , \\
   && \Phi(s,t) = - \sum_{i=0,j=1}^\infty \Phi_{i,-j} \, s^i t^{-j} \qquad \mbox{around } (s,t)=(0,\infty) \, , \\
   && \Phi(s,t) = \sum_{i,j=1}^\infty \Phi_{-i,-j} \, s^{-i} t^{-j} \qquad \mbox{around } (s,t)=(\infty,\infty) \, .
\end{eqnarray*}

\begin{proposition}
The following recurrence relation holds,
\begin{eqnarray}
     \bar{\mathrm{d}} \Phi_{i,j} = \mathrm{d} \Phi_{i,j+1} - (\mathrm{d} \Phi_{i,0}) \, \Phi_{0,j} = \mathrm{d} \Phi_{i+1,j} +  \Phi_{i,0} \, \mathrm{d} \Phi_{0,j} 
                   \qquad  i,j \in \mathbb{Z}  \, .       \label{Phi_first_order_recurrence_DMH}
\end{eqnarray}
\end{proposition}
\begin{proof}
By a direct computation we obtain
\begin{eqnarray*}
 \bar{\mathrm{d}} \Phi(s,t) + [ \mathrm{d} \Phi(s,0) ] \, \Phi(0,t) = t^{-1} \mathrm{d} [ \Phi(s,t) - \Phi(s,0) ] \, .
\end{eqnarray*}
Expansion leads to the asserted sequence of equations.
\end{proof}

Acting with $\mathrm{d}$ on (\ref{Phi_first_order_recurrence_DMH}) implies
\begin{eqnarray}
     \mathrm{d} \bar{\mathrm{d}} \Phi_{i,j} - \mathrm{d} \Phi_{i,0} \, \mathrm{d} \Phi_{0,j} = 0  \qquad \quad  i,j \in \mathbb{Z}  \, .   \label{Phi_eqs_recurrence_DMH}
\end{eqnarray}
In particular, setting $i=j=0$, we see that $\phi = - \Phi_{0,0}$ solves (\ref{phi_eq_DMH}), which is the binary Darboux transformation 
result for vanishing seed solution. 

\begin{example}
Choosing $i=0$ and $j=-1$ in (\ref{Phi_eqs_recurrence_DMH}), leads to 
\begin{eqnarray*}
     \bar{\mathrm{d}} (1 - \Phi_{0,-1}) + (\mathrm{d} \Phi_{0,0}) \, (1 - \Phi_{0,-1}) = 0 \, .
\end{eqnarray*}
As a consequence, if $g := (1-\Phi_{0,-1}) \, g_0$, where $g_0$ is $\mathrm{d}$- and $\bar{\mathrm{d}}$-constant, is invertible, then it solves (\ref{g_eq_DMH}). 
In a similar way, choosing $i=-1$ and $j=0$, we find that $h := h_0 \, (1 + \Phi_{-1,0})$, with
a $\mathrm{d}$- and $\bar{\mathrm{d}}$-constant $h_0$, solves $\mathrm{d}(h^{-1} \, \bar{\mathrm{d}} h) =0$, provided its inverse exists.  
\end{example}

\begin{lemma}
\begin{eqnarray*}
    \Phi_{i,-1} \, \Phi_{-1,j} = \Phi_{i-1,j} - \Phi_{i,j-1} \qquad \quad i,j \in \mathbb{Z}  \, .
\end{eqnarray*}
\end{lemma}
\begin{proof}
A direct computation shows that
\begin{eqnarray*}
    \Phi(s,0) \, \Phi(0,t) = t^{-1} [\Phi(s,t) - \Phi(s,0)] - s^{-1} [\Phi(s,t)-\Phi(0,t)]  \, .
\end{eqnarray*}    
Expansion yields
\begin{eqnarray*}
    \Phi_{i+1,j} - \Phi_{i,j+1} + \Phi_{i,0} \, \Phi_{0,j} = 0  \qquad \quad  i,j \in \mathbb{Z}  \, .
\end{eqnarray*}
By repeated application of this relation, we find
\begin{eqnarray*}
    \Phi_{i,-1} \, \Phi_{-1,j} &=& (\Phi_{i-1,0} - \Phi_{i-1,0} \, \Phi_{0,-1}) (\Phi_{0,j-1} + \Phi_{-1,0} \, \Phi_{0,j-1}) \\
    &=& \Phi_{i-1,0} \, \Phi_{0,j-1} + \Phi_{i-1,0} \, ( \Phi_{-1,0} - \Phi_{0,-1} - \Phi_{0,-1} \, \Phi_{-1,0} ) \, \Phi_{0,j-1} \\
    &=& \Phi_{i-1,j} - \Phi_{i,j-1} + \Phi_{i-1,0} \, (\Phi_{-1,0} \, \Phi_{0,-1} - \Phi_{0,-1} \, \Phi_{-1,0}) \, \Phi_{0,j-1} \\
    &=& \Phi_{i-1,j} - \Phi_{i,j-1} \, ,
\end{eqnarray*}
using $\Phi_{-1,0} \, \Phi_{0,-1} = \Phi_{0,-1} \, \Phi_{-1,0}$, which is a consequence of (\ref{Omega_Sylvester_DMH}).
\end{proof}

\begin{proposition}
\begin{eqnarray*}
    \bar{\mathrm{d}} \mathrm{d} \Phi_{i,j} = - (\bar{\mathrm{d}} \Phi_{i,-1}) \, \bar{\mathrm{d}} \Phi_{-1,j}  \qquad  \quad  i,j \in \mathbb{Z}  \, .
\end{eqnarray*}
\end{proposition}
\begin{proof} 
A special case of (\ref{Phi_first_order_recurrence_DMH}) is
\begin{eqnarray*}
    (\bar{\mathrm{d}} \Phi_{i,-1}) \, \Phi_{-1,j} = (\mathrm{d} \Phi_{i,0}) \, ( \Phi_{-1,j} - \Phi_{0,-1} \, \Phi_{-1,j} ) \, .
\end{eqnarray*}
Using the lemma, this becomes
\begin{eqnarray*}
    (\bar{\mathrm{d}} \Phi_{i,-1}) \, \Phi_{-1,j} = (\mathrm{d} \Phi_{i,0}) \, \Phi_{0,j-1} \, .
\end{eqnarray*}
 From (\ref{Phi_first_order_recurrence_DMH}) we now obtain
\begin{eqnarray*}
     \mathrm{d} \Phi_{i,j} 
   = \bar{\mathrm{d}} \Phi_{i,j-1} + (\mathrm{d} \Phi_{i,0}) \, \Phi_{0,j-1} 
   = \bar{\mathrm{d}} \Phi_{i,j-1} + (\bar{\mathrm{d}} \Phi_{i,-1}) \, \Phi_{-1,j} \, ,
\end{eqnarray*}
which implies our assertion by acting with $\bar{\mathrm{d}}$ on it.   
\end{proof}

Choosing $i=j=-1$ in the preceding proposition, we obtain the next result.

\begin{corollary}
$\Phi_{-1,-1}$ solves 
\begin{eqnarray*}
     \bar{\mathrm{d}} \mathrm{d} \varphi + \bar{\mathrm{d}} \varphi \, \bar{\mathrm{d}} \varphi = 0 \, ,
\end{eqnarray*}
which is (\ref{phi_eq_DMH}) with $\mathrm{d}$ and $\bar{\mathrm{d}}$ exchanged. 
\end{corollary}

\begin{remark}
The construction of equations in this section, and corresponding solutions, is 
analogous to the ``Cauchy matrix approach''\index{Cauchy matrix}, see \cite{Zhan+Zhao13_DMH,XZZ14_DMH,HJN16_DMH}.
\end{remark}

\section{Deformation of binary Darboux transformations and integrable systems with sources}
\label{sec:scs_DMH}
Let us replace $\Omega$ by $\Omega - \omega$ in the equations of Section~\ref{sec:bDt_DMH}, i.e., 
\begin{eqnarray*}
  &&  \Gamma \, (\Omega - \omega) - (\Omega - \omega) \, \Delta = \eta \, \theta \, , \\
  &&  \bar{\mathrm{d}} (\Omega - \omega) = \mathrm{d} (\Omega - \omega) \, \Delta - (\mathrm{d} \Gamma) \, (\Omega - \omega) 
     + \kappa \, (\Omega - \omega) + (\Omega - \omega) \, \lambda + (\mathrm{d} \eta) \, \theta  \, .
\end{eqnarray*}
Hence
\begin{eqnarray}
  &&  \Gamma \, \Omega - \Omega \, \Delta = \eta \, \theta + c \, , \nonumber  \\
  &&  \bar{\mathrm{d}} \Omega = (\mathrm{d} \Omega) \, \Delta - (\mathrm{d} \Gamma) \, \Omega 
     + \kappa \, \Omega + \Omega \, \lambda + (\mathrm{d} \eta) \, \theta + \gamma \, ,   \label{scs_Omega_eqs_DMH}
\end{eqnarray}
where
\begin{eqnarray}
     c := \Gamma \, \omega - \omega \, \Delta \, , \qquad
    \gamma := \bar{\mathrm{d}} \omega - (\mathrm{d} \omega) \, \Delta + (\mathrm{d} \Gamma) \, \omega - \kappa \, \omega 
              - \omega \, \lambda \, .
                                   \label{c,gamma_in_terms_of_omega_DMH}
\end{eqnarray}
We note that they satisfy 
\begin{eqnarray}
   && \bar{\mathrm{d}} \gamma = (\mathrm{d} \gamma) \, \Delta - (\mathrm{d} \Gamma) \, \gamma + \kappa \, \gamma - \gamma \, \lambda 
      - (\mathrm{d} \kappa) \, c \, , \nonumber \\
   && \bar{\mathrm{d}} c = (\mathrm{d} c) \, \Delta + \kappa \, c + c \, \lambda + \Gamma \, \gamma - \gamma \, \Delta \, .   
       \label{gamma,c_constraints_DMH}
\end{eqnarray}
By straightforward computations, one proves the following. 

\begin{theorem}[\cite{CDMH16_DMH}]
\label{thm:scs_DMH}
Let $\Delta, \Gamma, \kappa, \lambda$ satisfy (\ref{Delta,lambda_eqs_DMH}) and (\ref{Gamma,kappa_eqs_DMH}).
Let $\phi_0$ be a solution of (\ref{phi_eq_DMH}) and $\theta$, $\eta$, $\Omega$ satisfy the linear equations 
(\ref{theta_eq_DMH}), (\ref{eta_eq_DMH}) and (\ref{scs_Omega_eqs_DMH}), respectively.
Then
\begin{eqnarray}
    \phi = \phi_0 - \theta \, \Omega^{-1} \, \eta \, , \qquad
       q = \theta \, \Omega^{-1} \, , \qquad 
       r = \Omega^{-1} \, \eta  \, ,  \label{phi,q,r_DMH}
\end{eqnarray}
are solutions of
\begin{eqnarray}
    \mathrm{d} \, \bar{\mathrm{d}} \, \phi + \mathrm{d} \phi \; \mathrm{d} \phi = \mathrm{d} ( q \, \gamma \, r - q \, \mathrm{d} (c \, r) )     
                \label{scs_phi_eq_DMH}
\end{eqnarray}
and 
\begin{eqnarray}
    \bar{\mathrm{d}} q &=& (\mathrm{d} \phi) \, q + \mathrm{d}( q \, \Gamma) - q \, \kappa 
        - q \, \gamma \, \Omega^{-1} - (\mathrm{d} q) \, c \, \Omega^{-1} \, ,  \nonumber \\
    \bar{\mathrm{d}} r &=& - r \, \mathrm{d} \phi + \mathrm{d}( \Delta \, r) - \lambda \, r  
        - \Omega^{-1} \, \gamma \, r + \Omega^{-1} \, \mathrm{d}(c \, r) \, .       \label{scs_q,r_eqs_DMH}
\end{eqnarray}
\end{theorem}

\begin{remark}
The above equations form a consistent system in the sense that any equation derived from it by acting with 
$\bar{\mathrm{d}}$ on any of its members yields an equation that is satisfied as a consequence of the system. 
\end{remark}

\begin{remark}
\label{rem:phi,q,r,hatOmega_sys_DMH}
Using the last two equations in (\ref{phi,q,r_DMH}), we eliminate $\theta$ and $\eta$ in (\ref{scs_Omega_eqs_DMH}) 
to obtain
\begin{eqnarray*}
  &&  \Delta \, \hat{\Omega} - \hat{\Omega} \, \Gamma = r \, q + \hat{\Omega} \, c \, \hat{\Omega} \, , \\
  &&  \bar{\mathrm{d}} \hat{\Omega} = \mathrm{d} (\hat{\Omega} \Gamma) - \hat{\Omega} \, \kappa - \lambda \, \hat{\Omega} + (\mathrm{d} r) \, q
                        + \left( (\mathrm{d} \hat{\Omega}) \, c + \hat{\Omega} \, \gamma \right) \, \hat{\Omega} \, ,
\end{eqnarray*}
where $\hat{\Omega} = - \Omega^{-1}$. (\ref{scs_q,r_eqs_DMH}) then reads
\begin{eqnarray*}
    \bar{\mathrm{d}} q &=& (\mathrm{d} \phi) \, q + \mathrm{d}( q \, \Gamma) - q \, \kappa 
        + q \, \gamma \, \hat{\Omega} + (\mathrm{d} q) \, c \, \hat{\Omega} \, ,  \nonumber \\
    \bar{\mathrm{d}} r &=& - r \, \mathrm{d} \phi + \mathrm{d}( \Delta \, r) - \lambda \, r  
        + \hat{\Omega} \, \gamma \, r - \hat{\Omega} \, \mathrm{d}(c \, r) \, . 
\end{eqnarray*}
Together with (\ref{scs_phi_eq_DMH}), this constitutes a system of equations for $\phi$, $q$, $r$ and $\hat{\Omega}$, for 
which we now have a solution-generating method at hand. 
\end{remark}

\begin{example}
\label{ex:KdV_bidiff_DMH}
Let $\mathcal{A}_0$ be the space of smooth complex functions on $\mathbb{R}^2$. We extend it to $\mathcal{A} = \mathcal{A}_0[\partial]$, where 
$\partial$ is the partial differentiation operator with respect to the coordinate $x$. On $\mathcal{A}$ we define  
\begin{eqnarray*}
    \mathrm{d} f = [\partial, f] \, \xi_1 + \frac{1}{2} \, [\partial^2,f] \, \xi_2 \, , \qquad
    \bar{\mathrm{d}} f = -\frac{1}{2} \, [\partial^2,f] \, \xi_1
            + \frac{1}{3} \, [\partial_t - \partial^3,f] \, \xi_2  
\end{eqnarray*}
(cf. \cite{DMH08bidiff_DMH}). The maps $\mathrm{d}$ and $\bar{\mathrm{d}}$ extend to linear maps on 
$\boldsymbol{\Omega} = \mathcal{A} \otimes \bigwedge \mathbb{C}^2$. 
Choosing 
\begin{eqnarray*}
      \Delta = \Gamma = -\partial \, ,   \quad 
      \kappa = \frac{1}{2} \, Q^2 \, (\xi_1 + \partial \, \xi_2) \, ,   \quad 
      \lambda = -\frac{1}{2} \, P^2 \, (\xi_1 + \partial \, \xi_2) \, ,
\end{eqnarray*} 
with constant matrices $P,Q$, (\ref{Delta,lambda_eqs_DMH}) and (\ref{Gamma,kappa_eqs_DMH}) are satisfied. Requiring $\omega_x =0$ means $c=0$. 
Furthermore, we are led to set
\begin{eqnarray*}
      \gamma = \gamma_1 \, \xi_1 + (\gamma_2 + \gamma_1 \, \partial) \, \xi_2 \, , \qquad
      \gamma_1 = - \frac{1}{2} (Q^2 \omega - \omega \, P^2) \, , \quad
      \gamma_2 = \frac{1}{3} \, \omega_t \, .
\end{eqnarray*}
\textbf{1.} Let $\omega$ be constant, so that $\gamma_2=0$. Then, in terms of $u = 2 \phi_x$, 
(\ref{scs_phi_eq_DMH}) and (\ref{scs_q,r_eqs_DMH}) lead 
to\footnote{There is a wrong factor in front of the cubic nonlinearities in (4.5) of \cite{CDMH16_DMH}.} 
\begin{eqnarray*}
   4 \, u_t - u_{xxx} - 3 \, (u^2)_x  &=& \left( q_x \, r - q \, r_x \right)_x \, , \label{KdV_scs_u_DMH} \\
   q_t - q_{xxx} - \frac{1}{4} \, q \, r \, q &=& \frac{3}{4} \, ( u_x \, q + 2 \, u \, q_x )  \, , \\
   r_t - r_{xxx} + \frac{1}{4} \, r \, q \, r &=& \frac{3}{4} \, ( r \, u_x + 2 \, r_x \, u )  \, , 
\end{eqnarray*}
after a redefinition of $q$. Here we disregarded those equations resulting from (\ref{scs_q,r_eqs_DMH}) 
that involve $\gamma_1$. 
The above system is the second member of the Yajima-Oikawa hierarchy\index{Yajima-Oikawa hierarchy} 
(also see the appendix of \cite{CDMH16_DMH}).\footnote{We are grateful to Dmitry Demskoi and Maxim Pavlov 
for informing us about this.} 
The first equation is a KdV equation\index{Korteweg-deVries equation} with ``sources''. \\
\textbf{2.} Setting $\gamma_1 =0$, we are led to
\begin{eqnarray*}
   4 \, u_t - u_{xxx} - 3 \, (u^2)_x  
     = 8 \, (q \, \omega_t \, r)_x  \, ,  \quad
  q_{xx} = q \, Q^2 - u \, q  \, , \quad
  r_{xx} = P^2 \, r - r \, u \, .
\end{eqnarray*}
where $\omega_t$ can be absorbed by a redefinition of $q$ or $r$. Here we disregarded equations 
resulting from (\ref{scs_q,r_eqs_DMH}) 
that involve $\gamma_2$. This example of a ``system with self-consistent sources'' appeared in \cite{Mel'88_DMH}. \\
In both cases, Theorem~\ref{thm:scs_DMH} generates exact solutions. Soliton solutions are obtained starting 
with vanishing seed solution. Many further examples can be found in \cite{CDMH16_DMH}.
\end{example}

\section{Conclusions and further remarks}
\label{sec:conclusions_DMH}
Bidifferential calculus\index{bidifferential calculus} is a drastic abstraction of structures that are relevant in the theory of 
(completely) integrable PDDEs\index{PDDE! integrable}. This structure may well find applications far away from the latter. 
This suggests to look for associative algebras\index{algebra! associative} that admit a bidifferential calculus. So far, 
there are only few examples beyond differential geometry\index{geometry! differential} (Fr\"olicher-Nijenhuis 
theory\index{Fr\"olicher-Nijenhuis theory}) and examples with a graded algebra\index{algebra! graded} of the 
form (\ref{Omega_wedge_DMH}), which underlies most of our work on  
integrable PDDEs\index{PDDE! integrable}. For example, in \cite{Sita00_DMH} a bidifferential 
calculus\index{bidifferential calculus} has been found 
for the quantum group\index{quantum group} $U_q(\mathrm{sl}(2))$, with $q$ a third root of unity. There is a vast literature 
by now on differential calculi on some classes of associative algebras\index{algebra! associative}. For any given 
$(\boldsymbol{\Omega},\mathrm{d})$, the question is then whether there exists another graded 
derivation\index{derivation! graded} $\bar{\mathrm{d}}$ 
on $\boldsymbol{\Omega}$ that extends it to a bidifferential caclulus\index{bidifferential calculus}. For instance, this problem 
could be addressed for (bicovariant) differential calculi on Hopf algebras\index{algebra! Hopf} (quantum groups\index{quantum group}). 
Of particular interest are also deformations of algebras that underly classical integrable systems, 
leading to a kind of ``quantization'' of the latter. What we are looking for, in particular, is a 
noncommutative version of Fr\"olicher-Nijenhuis theory\index{Fr\"olicher-Nijenhuis theory}. 

Even concerning classical integrable systems, bidifferential calculus\index{bidifferential calculus} has not been explored 
systematically, so far. 
Given a bidifferential calculus, equations like (\ref{phi_eq_DMH}) or (\ref{g_eq_DMH}) are not necessarily 
equivalent to a non-trivial PDDE\index{PDDE} (or a system of PDDEs). How to express the dependent variable 
in a suitable way in terms of the non-commuting elements of the algebra $\mathcal{A}$ ? So far this is still 
too much based on trial and error, and more systematics would be desirable. 

We mention the following in order to build a bridge to other contributions to these proceedings, 
dealing with deformed derivations\index{derivation! deformed}. Let the space of 1-forms 
$\boldsymbol{\Omega}^1$ admit a finite left and also right $\mathcal{A}$-module\index{module} basis $\theta^s$, $s=1,\ldots,S$. 
Then there are maps $\Theta^s_{s'} : \mathcal{A} \to \mathcal{A}$ such that
\begin{eqnarray*}
    f \, \theta^s = \sum_{s'=1}^S \theta^{s'} \, \Theta^s_{s'}(f)  \, ,
\end{eqnarray*}
for all $f \in \mathcal{A}$. Introducing generalized (left and right) partial derivatives\index{partial derivative! generalized} 
via
\begin{eqnarray*}
     \mathrm{d} f = \sum_{s=1}^S (\partial^{(L)}_s f) \, \theta^s = \sum_{s=1}^S \theta^s \, (\partial^{(R)}_s f) \, ,
\end{eqnarray*}
the Leibniz rule yields
\begin{eqnarray*}
     \sum_{i=1}^J \theta^s \, \partial^{(R)}_s (f h)  
   &=& \mathrm{d} (f h) 
   = (\mathrm{d} f) \, h + f \, \mathrm{d} h \\
   &=& \sum_{s=1}^S \theta^s \, (\partial^{(R)}_s f) \, h 
     + f \, \sum_{s'=1}^S \theta^{s'} \, \partial^{(R)}_{s'} h \\
   &=& \sum_{s=1}^S \theta^s \, [ \partial^{(R)}_s f) \, h + \sum_{s'=1}^S \Theta^{s'}_s(f) \, \partial^{(R)}_{s'} h ]   
\end{eqnarray*}
Hence the right partial derivatives satisfy the twisted Leibniz rule 
\begin{eqnarray*}
    \partial^{(R)}_s (f h) = (\partial^{(R)}_s f) \, h + \sum_{s'=1}^S \Theta^{s'}_s(f) \, \partial^{(R)}_{s'} h \, .
\end{eqnarray*}
If $\Theta$ is diagonal, then the maps $\Theta^s_s$ are automorphisms\index{automorphism} and $\partial^{(R)}_s$ is a 
generalized derivation\index{derivation! generalized}
\cite{March88_DMH,Bout+Marc94_DMH} (called $\sigma$-derivation in \cite{HLS06_DMH}, for example). Also see 
\cite{DMH04auto_DMH,DMH04CJP_DMH} for corresponding examples. 
There is a similar formula for the left derivatives, of course, and corresponding generalized 
partial derivatives\index{partial derivative! generalized} can also be associated with $\bar{\mathrm{d}}$. 
If the aforementioned conditions are met, equations considered in this work can be expressed in terms of these 
generalized partial derivatives\index{partial derivative! generalized}.

\providecommand{\bysame}{\leavevmode\hbox to3em{\hrulefill}\thinspace}

\end{document}